\documentclass{article}
\usepackage{ijcai21}

\usepackage{times}
\usepackage{soul}
\usepackage{url}
\usepackage[hidelinks]{hyperref}
\usepackage[utf8]{inputenc}
\usepackage[small]{caption}
\usepackage{graphicx}
\usepackage{amsmath}
\usepackage{amsthm}
\usepackage{booktabs}
\usepackage{algorithm}
\usepackage{latexsym}
\usepackage{soul}
\usepackage[utf8]{inputenc}
\usepackage{amsmath}
\usepackage{booktabs}
\usepackage{amsthm}
\usepackage{thmtools}
\usepackage{thm-restate}

\usepackage[utf8]{inputenc} 
\usepackage[T1]{fontenc}    
\usepackage{url}            
\usepackage{booktabs}       
\usepackage{amsfonts}       
\usepackage{nicefrac}       
\usepackage{microtype}      
\usepackage{color}

\newtheorem{definition}{Definition}

\usepackage{helvet}  
\usepackage{courier}  
\usepackage{subfig}
\usepackage{enumitem}
\usepackage{multirow}
\usepackage{mathtools}
\usepackage{balance}
\usepackage{flushend}

\graphicspath{{figures/}}

\newcommand{\cD}{\mathcal{D}}

\newcommand{\cR}{\mathcal{R}}
\newcommand{\cS}{\mathcal{S}}

\newcommand{\cM}{\mathcal{M}}

\usepackage{amssymb}
\usepackage{tabularx}
\usepackage{hhline}
\usepackage{algpseudocode}
\usepackage{array}
\usepackage{listings}
\usepackage{diagbox}
\usepackage{tikz}
\usepackage[utf8]{inputenc}
\usepackage{bm}
\usetikzlibrary{spy}

\usepackage{pgfplots}
\pgfplotsset{compat=newest}
\usepgfplotslibrary{groupplots}
\usepgfplotslibrary{dateplot}

\input
\usepackage{makecell}
\makeatletter
\def\BState{\State\hskip-\ALG@thistlm}
\makeatother
\usepackage{setspace}

\makeatletter
\def\Cline#1#2{\@Cline#1#2\@nil}
\def\@Cline#1-#2#3\@nil{%
  \omit
  \@multicnt#1%
  \advance\@multispan\m@ne
  \ifnum\@multicnt=\@ne\@firstofone{&\omit}\fi
  \@multicnt#2%
  \advance\@multicnt-#1%
  \advance\@multispan\@ne
  \leaders\hrule\@height#3\hfill
  \cr}
\makeatother

\makeatletter
\def\BState{\State\hskip-\ALG@thistlm}
\makeatother
\usepackage{colortbl}

\newcommand{\fade}[1]{\textcolor{gray}{#1}}
\usepackage[switch]{lineno}

\newcommand*\system{\textsc{FedMD-NFDP}}
\newcommand*\LDP{\textsc{FedMD-LDP}}
\newcommand*\PATE{\textsc{PATE}}
\newcommand*\DPSGD{\textsc{DP-SGD}}
\title{Federated Model Distillation with Noise-Free Differential Privacy}

\author{
Lichao Sun$\textsuperscript{\rm 1}$
\and
Lingjuan Lyu$\textsuperscript{\rm 2}$\thanks{Equal contribution. Order determined by coin toss.}
\affiliations
$\textsuperscript{\rm 1}$ Lehigh University, $\textsuperscript{\rm 2}$ Ant Group
\emails
lis221@lehigh.edu, lingjuanlvsmile@gmail.com
}
\begin{document}

\maketitle
\begin{abstract}
Conventional federated learning directly averages model weights, which is only possible 
for collaboration between models with homogeneous architectures. Sharing prediction instead of weight removes this obstacle and eliminates the risk of white-box inference attacks in conventional federated learning. However, the predictions from local models are sensitive and would leak training data privacy to the public. 
To address this issue, one naive approach is adding the differentially private random noise to the predictions, which however brings a substantial trade-off between privacy budget and model performance. 
In this paper, we 
propose a novel framework called \system, which applies a Noise-Free Differential Privacy (NFDP) mechanism into a federated model distillation framework. 
Our extensive experimental results on various datasets validate that \system\ can deliver not only comparable utility and communication efficiency but also provide a noise-free differential privacy guarantee. We also demonstrate the feasibility of our \system\ by considering both IID and non-IID setting, heterogeneous model architectures, and unlabelled public datasets from a different distribution. 
\end{abstract}

\section{Introduction}
Federated learning (FL) provides a privacy-aware paradigm of model training, which allows a multitude of parties to construct a joint model without directly exposing their private training data~\cite{mcmahan2017communication,bonawitz2017practical,xu2021fedmood}. 
Nevertheless, recent works have demonstrated that FL may not always provide sufficient privacy guarantees, as communicating model updates throughout the training process can nonetheless reveal sensitive information~\cite{bhowmick2018protection,melis2019exploiting}.

In order to protect training data privacy in FL, various privacy protection techniques have been proposed in the literature~\cite{geyer2017differentially,mcmahan2018learning,bonawitz2017practical,wang2019collecting,zhao2020local,sun2020ldp,lyu2020privacy}. From the perspective of differential privacy, most works focus on the centralized differential privacy (CDP) that requires a central trusted party to add noise to the aggregated gradients~\cite{geyer2017differentially,mcmahan2018learning}. Moreover, these works are geared to tackle thousands of users for training to converge and achieve an acceptable trade-off between privacy and accuracy~\cite{mcmahan2018learning}, resulting in a convergence problem with a small number of parties.

To achieve stronger privacy protection, a few recent works start to integrate local differential privacy (LDP) into federated learning. 
However, most existing approaches 
can only support shallow models such as logistic regression and only focus on simple tasks and datasets~\cite{wang2019collecting,zhao2020local}. \cite{bhowmick2018protection} presented a viable approach to large-scale local private model training. Due to the high variance of their mechanism, it requires more than 200 communication rounds and incurs much higher privacy cost, i.e., MNIST ($\epsilon = 500$) and CIFAR-10 ($\epsilon = 5000$). A recent work~\cite{sun2020ldp} utilized Local Differential Privacy (LDP) into federated learning. However, in order to achieve a reasonable privacy budget, it uses a splitting and shuffling mechanism that split all parameters of a single model and send them individually to the cloud. This special communication requires tons of communication between clients and clouds. 

\begin{figure*}[!htp]
\centering
\subfloat[\LDP]{\includegraphics[height=1.7in,width=3.37in]{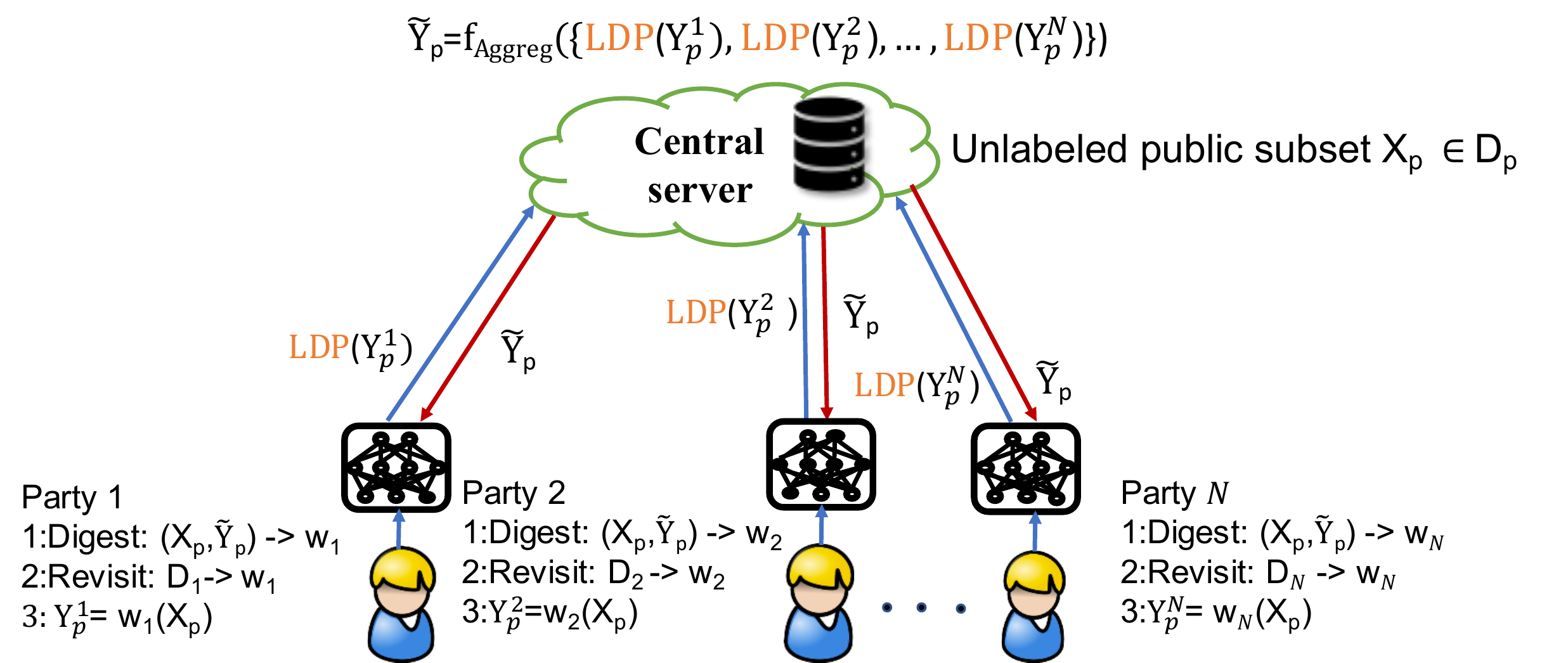}} \ \ \
\subfloat[\system]{\includegraphics[height=1.7in,width=3.37in]{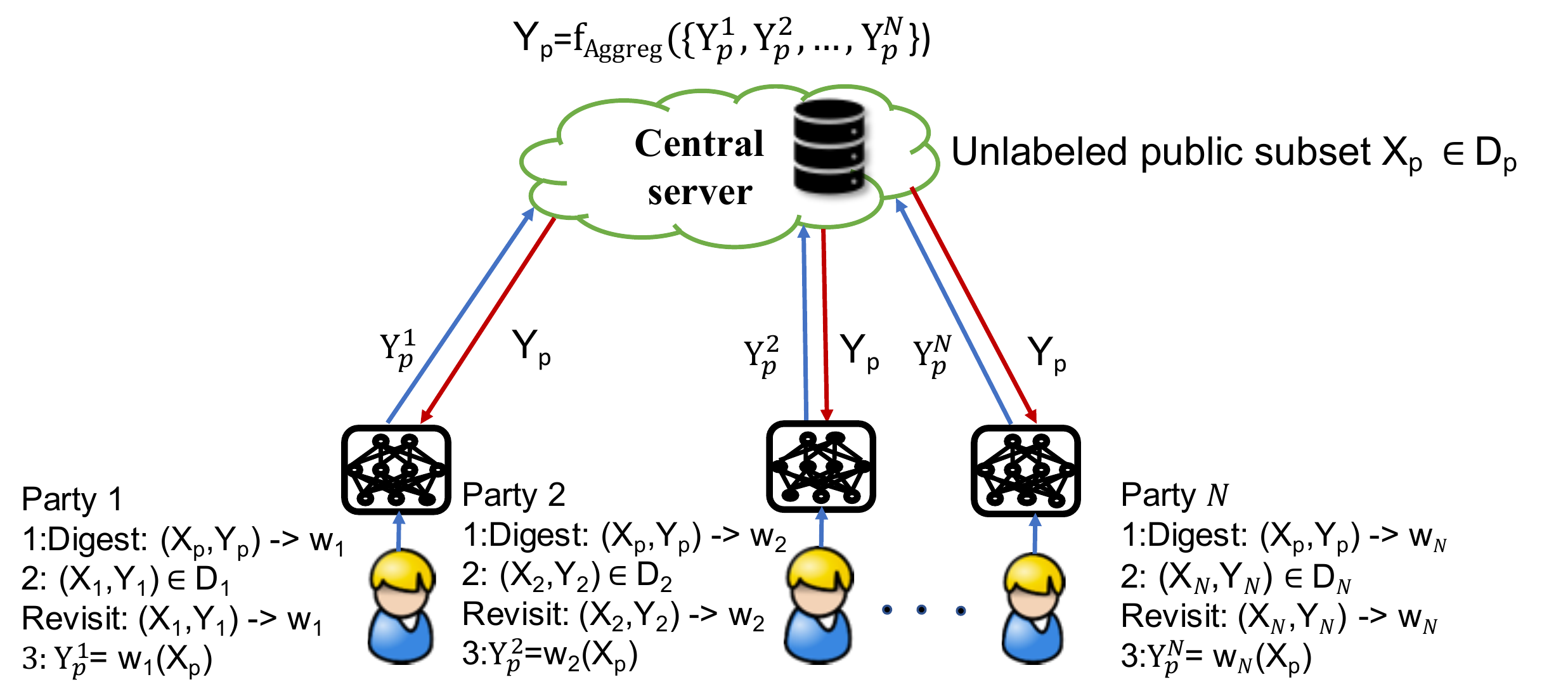}}
\caption{Overview of the naive \LDP\ and our \system\ in each communication round. Each party first updates its model $w_i$ to approach the consensus on the public dataset (Digest); then updates its model $w_i$ on its own sampled subset from private local data (Revisit). Note that the sampled subset $(X_i,Y_i) \in D_i$ is fixed in all rounds.}\label{fig:CFL_FedMD}
\end{figure*}

All the above works considered privacy issues in conventional FL that requires parties to share model weights. Compared with sharing prediction via knowledge transfer, conventional FL system~\cite{mcmahan2017communication,mcmahan2018learning} suffers from several intrinsic limitations: 
(1) it requires every party to share their local model weights in each round, thus limiting only to models with homogeneous architectures; 
(2) sharing model weight incurs a significant privacy issue of local model, as it opens all the internal state of the model to white-box inference attacks; 
(3) model weight is usually of much higher dimension than model predictions, resulting in huge communication overhead and higher privacy cost. 

Inspired by the knowledge transfer algorithms~\cite{bucilu2006model,hinton2015distilling}, \emph{Federated Model Distillation} (FedMD) shares the knowledge of FL parties' models via their predictions on an unlabeled public set \cite{li2019fedmd}. However, sharing prediction may still leak the privacy of the local data~\cite{papernot2017semi}. Currently, there is no reasonable privacy guarantee for sharing model prediction in FL. A naive approach is 
to add the differentially private random noise perturbation to the predictions of local models. However, prior works have shown the significant trade-off between privacy budget and model performance \cite{bhowmick2018protection}. 
We fill in the above gaps and make the following contributions:
\begin{itemize}
    \item We propose \system, a novel federated model distillation framework with the new proposed noise-free differential privacy (NFDP) mechanism that guarantees each party's privacy without explicitly adding any noise. 
    \item We formally prove that NFDP with both replacement and without replacement sampling strategies can inherently ensure $(\epsilon,\delta)$-differential privacy, eliminating noise addition and privacy cost explosion issues explicitly in previous works.
    \item Extensive experiments on benchmark datasets, various settings (IID and non-IID data distribution), and heterogeneous model architectures, demonstrate that \system\ achieves comparable utility with only a few private samples that are randomly sampled from each party, validating the numerous benefits of our framework.
\end{itemize}
We remark that sampling can individually serve as a privacy amplification method to tighten the privacy budget of a differentially private algorithm \cite{balle2018privacy}, which is
in sharp contrast with the inherent privacy guarantee of sampling given in this paper. 

\section{Preliminary}
\paragraph{Differential Privacy.}
DP has become a de facto standard for privacy analysis. DP can either be enforced in a ``local" or ``global" sense depending on whether the server is trusted. For FL scenarios where data are sourced from multiple parties, while the server is untrusted, DP should be enforced in a ``local" manner to enable parties to 
apply DP mechanisms before data publication, which we term as LDP. Compared with the global model via DP (CDP)~\cite{dwork2014algorithmic,abadi2016deep}, LDP offers a stronger level of protection. 

\begin{definition}\label{def:dp}
A randomized mechanism $\cM$: $\cD \to \cR$ with domain $\cD$ and range $\cR$ satisfies $(\epsilon,\delta)$-differential privacy if for all two neighbouring inputs $D,D'\in\cD$ and any measurable subset of outputs $S \subseteq \cR$ it holds that
\begin{eqnarray*}
\Pr\{\cM(D)\in S\} &\leq& \exp(\epsilon) \cdot \Pr\{\cM(D')\in S\} +
\delta\enspace.
\end{eqnarray*}
\end{definition}
A formal definition of record-level DP is provided in Def.~\ref{def:dp}, which bounds the effect of the presence or the absence of a record on the output likelihood within a small factor $\epsilon$. The additive term $\delta$ allows that the unlikely responses do not need to satisfy the pure $\epsilon$-DP criterion. In FL, each party can individually apply $\cM$ in Definition~\ref{def:dp} to ensure record-level LDP. 

\section{Federated Model Distillation with Noise-Free Differential Privacy}

Unlike the existing federated learning algorithms, such as FedAvg~\cite{mcmahan2017communication}, FedMD does not force a single global model onto local models. Instead, each local model is updated separately. To support heterogeneous model architectures, we assume that an unlabeled public dataset is available, then parties share the knowledge that they have learned from their training data (their model predictions) in a succinct, black-box and model agnostic manner. To protect local model predictions, each party can explicitly apply LDP mechanisms by adding noise to their local model predictions, as shown in \LDP\ (Figure~\ref{fig:CFL_FedMD} (a)), or adopt data sampling before training, which inherently ensures LDP of the sampled subset, and the follow-up local model predictions as per the post-processing property of DP~\cite{dwork2014algorithmic}, as demonstrated in \system (Figure~\ref{fig:CFL_FedMD} (b)). 

It should be noted that \LDP\ requires each party to explicitly inject noise to ensure DP individually before releasing their local model knowledge to the server. The privacy cost will accumulate as per the dimension of the shared knowledge ($|Y_p|*class$), as well as the communication rounds, resulting in huge privacy costs. Here $|Y_p|$ is the number of the chosen public set and $class$ refers to the class number. In contrast, our \system\ inherently ensures that the released local model knowledge by each party is differentially private via random data sampling process, as indicated in Theorem~\ref{theorem_sampling1} and Theorem~\ref{theorem_sampling2}.

\begin{algorithm}[t]
\caption{\system.~Initialization phase does not involve collaboration. $D_i$ and $w_i$ are local dataset and model parameters from $i$-th party. $Y_p^i[t]$ is the prediction from $i$-th party on the chosen public subset $X_p \in D_p$ in round $t$.}
\label{alg:FedMD-NF}
\begin{algorithmic}[1]
\State \textbf{\hskip 8em Initialization phase}
\State Initializes each party $i \in [N]$ with the same pretrained model, selects subset $(X_i,Y_i) \in D_i$, and updates their weights in parallel:
	\For{$t\in[T_1]$ epochs} 
	\State Update $w_i \leftarrow$
	 \textsc{\textsf{\small{Train}}} $(w_i,X_i,Y_i)$
	\EndFor
	\State Server randomly samples a public subset $X_p[0] \in D_p$
	\State Send $Y_p^{i}[0] = \textsc{\textsf{\small{Predict}}}(w_i;X_p[0])$ to the server
\vspace{1em}
\State\textbf{\hskip 8em Collaboration phase}
\vspace{.3em}
\State $Y_p[0] = f_\mathsf{Aggreg}(\{Y_p^{i\in[N]}[0]\})$
\Comment{\fade{Initial aggregation at the server}}

\For{$t\in[R]$ communication rounds}
    \State Server randomly samples a public subset $X_p[t+1] \in D_p$
	\For{$i\in[N]$ parties} \Comment {\fade{Each party updates local weight $w_i$ in parallel}}
	    \For{$j\in[T_2]$ epochs} 
        \State Digest: $w_i \leftarrow$ \textsc{\textsf{\small{Train}}} $(w_i,X_p[t],Y_p[t])$
        \EndFor
        
        \For{$j\in[T_3]$ epochs} 
		\State Revisit: $w_i \leftarrow$ \textsc{\textsf{\small{Train}}} $(w_i,X_i,Y_i)$
        \EndFor
        
		\State Send $Y_p^{i}[t+1] = \textsc{\textsf{\small{Predict}}}(w_i;X_p[t+1])$ to the server 

        
	\EndFor
\State $Y_p[t+1] = f_\mathsf{Aggreg}(\{Y_p^{i\in[N]}[t+1]\})$
\Comment{\fade{Prediction aggregation at the server}}
\EndFor
\end{algorithmic}
\end{algorithm}

Algorithm~\ref{alg:FedMD-NF} describes our \system\ algorithm, which consists of two training phases: (1) during initialization phase, every party $i$ updates its local model weights $w_i$ on a randomly sampled subset $(X_i,Y_i) \in D_i$ from local private training data $D_i$ for $T_1$ times without any collaboration; (2) during collaboration phase, parties share the knowledge of their local models via their predictions on a subset of public data, $X_p$. In each round of the collaboration phase, the detailed procedure 
proceeds as follows:

\begin{itemize}
\item Each party uses its local model weights $w_i$ to compute prediction $Y_p^i$ for $X_p$ and shares them with the server. 
\item The server aggregates the predictions (separately for each public record), i.e., computes $Y_p=f_\mathsf{Aggreg}(Y_p^1,\cdots,Y_p^N)$, and sends $Y_p$ to all parties for the next round's local training; $f_\mathsf{Aggreg}$ is an aggregation algorithm, which is average function throughout this work.
\item Each party first updates its local model weights $w_i$ by training on the soft-labeled public data $(X_p,Y_p)$ to approach the consensus on the public dataset (Digest); then training on its 
previously sampled local subset 
(Revisit).
\end{itemize}
 
In addition, Algorithm \ref{alg:FedMD-NF} can also support the implementation of \LDP. The only difference is the output $Y_p^{i}[t+1]$ (line 19) should be perturbed by the differentirally private random noise, which can be randomly sampled from either Laplace or Gaussian distribution.

\section{Theoretical Analysis}
In this work, we consider record-level DP for each party.
Below, we formally stated that NFDP with random sampling from each party's training dataset 
satisfies differential privacy guarantee for each party. 
In particular, random sampling without replacement and with replacement are two most common sampling strategies, and we prove the $(\epsilon,\delta)$-differential privacy for both of them.

\begin{restatable}{theorem}{withoutreplacement}
[NFDP mechanism: $(\epsilon,\delta)$-differential privacy of sampling {without replacement}]
\label{theorem_sampling1}
Given a training dataset of size $n$, sampling without replacement achieves $(\ln{\frac{n+1}{n+1-k}},\frac{k}{n})$-differential privacy, where $k$ is the subsample size.
\end{restatable}

\begin{restatable}{theorem}{withreplacement}
[NFDP mechanism: $(\epsilon,\delta)$-differential privacy of sampling {with replacement}]
\label{theorem_sampling2}
Given a training dataset of size $n$, sampling with replacement achieves $((k\ln{\frac{n+1}{n}}, 1-\left(\frac{n-1}{n}\right)^k)$-differential privacy, where $k$ is the subsample size.
\end{restatable}

\begin{restatable}{lemma}{compare}
Algorithm \ref{alg:FedMD-NF} using sampling with replacement is consistently more private than using sampling without replacement for any $n > 0$ and $0 < k \leq n$.
\label{lemma_compare}
\end{restatable}

All the related proofs of lemma and theorems can be referred to the Appendix. The nice property of NFDP with random sampling once and the post-processing property of differential privacy~\cite{dwork2014algorithmic} removes the privacy dependence on the number of queries on the public dataset, allowing a more practical deployment of our NFDP in FEDMD.

\begin{figure*}[t]
\centering
\subfloat[$\epsilon=k\ln{\frac{n/N+1}{n/N}}$]{\includegraphics[width=2.2in]{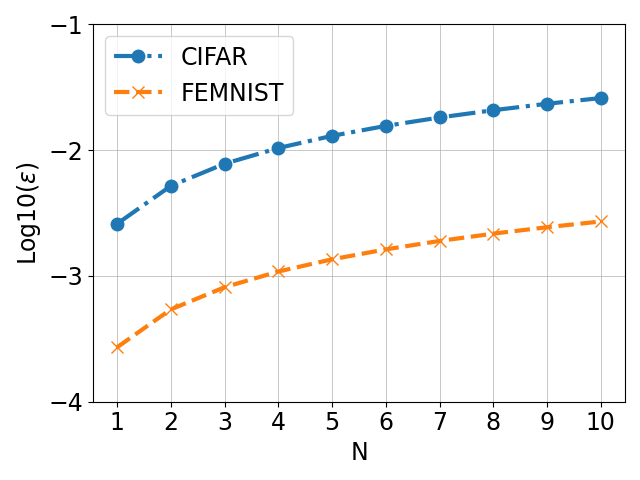}}
\subfloat[$\delta=1-\left(\frac{n/N-1}{n/N}\right)^k$]{\includegraphics[width=2.2in]{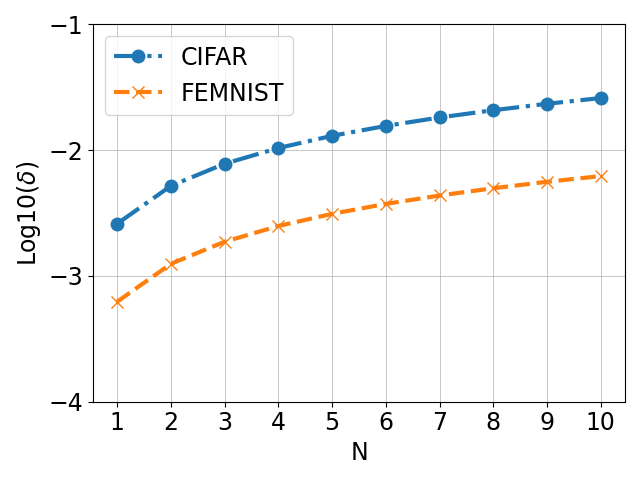}}\\
\subfloat[$\sigma=c_2\frac{\sqrt{T\log(1/\delta)}}{\epsilon}$]{\includegraphics[width=2.2in]{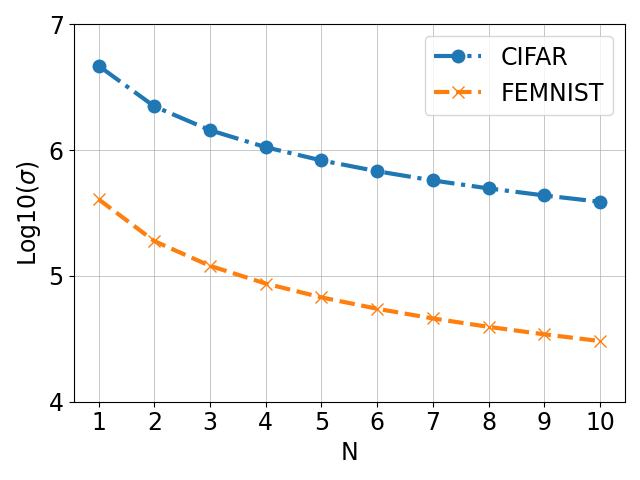}}
\subfloat[Accuracy v.s. Number of Parities]{\includegraphics[width=2.2in]{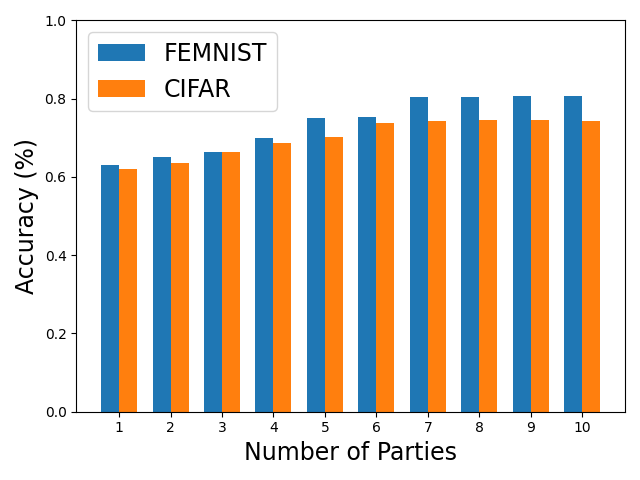}}
\caption{[a-c]: Log 10 scale of $\epsilon$, $\delta$ for each party in \system/\LDP\ v.s. the number of parities $N$. Note that, we use fixed $k=3$. Here $T$ is the total number of queries and we have $T=100000$ in 20 round communications, i.e., ~5000 queries on the public dataset in each round. $n$ is the size of private datasets owned by all parties, which is 28800 for FEMNIST and 3000 for CIFAR-10. The log 10 scale of $\delta$ of \LDP\ is calculated by $\epsilon$ and $\delta$, which are from (a) and (b) by the corresponding $N$. [d]: accuracy v.s. number of parties.} \label{fig:parties}
\end{figure*}

\section{Experimental Evaluation}

\begin{table}[b]
  
  \centering
\resizebox{3.350in}{!}{%
\begin{tabular}{c|c|c}
    \hline
 Task & Public & Private \\
 \hline
 IID &MNIST  & FEMNIST letters [a-f] classes \\
 \hline
 IID   & CIFAR-10 & CIFAR-100 subclasses [0,2,20,63,71,82] \\ \hline
 Non-IID& MNIST & FEMNIST letters from one writer  \\\hline
 Non-IID  &  CIFAR-10 & CIFAR-100 superclasses [0-5]  \\
 \hline
\end{tabular}
}
\caption{Summary of datasets} \label{summary-datasets}
\label{table:data_summary}
\end{table}

In the experiment, we evaluate on paired datasets, i.e., MNIST/FEDMNIST and CIFAR-10/CIFAR-100. For MNIST/FEMNIST, the public data is the MNIST, and the private data is a subset of the Federated Extended MNIST (FEMNIST)~\cite{caldas2018leaf}, which is built by partitioning the data in Extended MNIST based on the writer of the digit/character. In the IID scenario, the private dataset of each party is drawn randomly from FEMNIST. In the non-IID scenario,  each party only has letters written by a single writer, and the task is to classify letters by all writers.
 
For CIFAR-10/CIFAR-100, the public dataset is the CIFAR-10, and the private dataset is a subset of the CIFAR-100~\cite{krizhevsky2009learning}, which has 100 subclasses that fall under 20 superclasses, e.g., bear, leopard, lion, tiger, and wolf belong to large carnivores~\cite{li2019fedmd}. In the IID scenario, each party is required to classify test images into correct subclasses. The non-IID scenario is more challenging: each party has data from one subclass per superclass but needs to classify generic test data into the correct superclasses.
Therefore, it necessitates knowledge sharing among parties.

Each party's local model is two or three-layer deep neural networks for both MNIST/FEMNIST and CIFAR-10/CIFAR-100.
All experiments are implemented by using Pytorch.
A single GPU NVIDIA Tesla V100 is used in the experiments.
FEMNSIT and CIFAR-10 can be done within an hour at $N=10$ parties.
A summary of the public and private datasets used in this paper is provided in Table~\ref{summary-datasets}.

In each communication round, we use a subset of size 5000 that is randomly selected from the entire public dataset. We empirically validate \system\ largely reduces the communication cost without degrading the utility.
The number of training epochs in Algorithm~\ref{alg:FedMD-NF} and the batch size in the Digest and the Revisit phase may impact the stability of the learning process. We empirically choose $R=20, T_1=20, T_2=2, T_3=1$ via grid search. We initialize all parties with the same pre-trained model on some labelled data in the same domain. For example, parties training on the private FEMNIST are initialized with the same pre-trained model on some labelled MNIST data.

\begin{figure*}[t]
\centering
\subfloat[FEMNIST]{\includegraphics[width=2.9in]{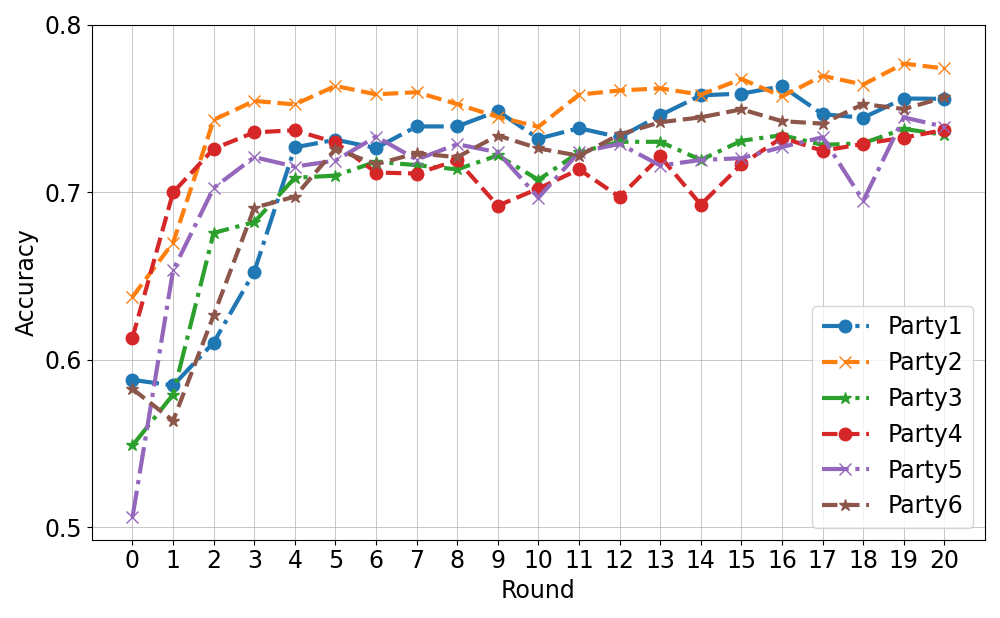}}
\subfloat[CIFAR-10]{\includegraphics[width=2.9in]{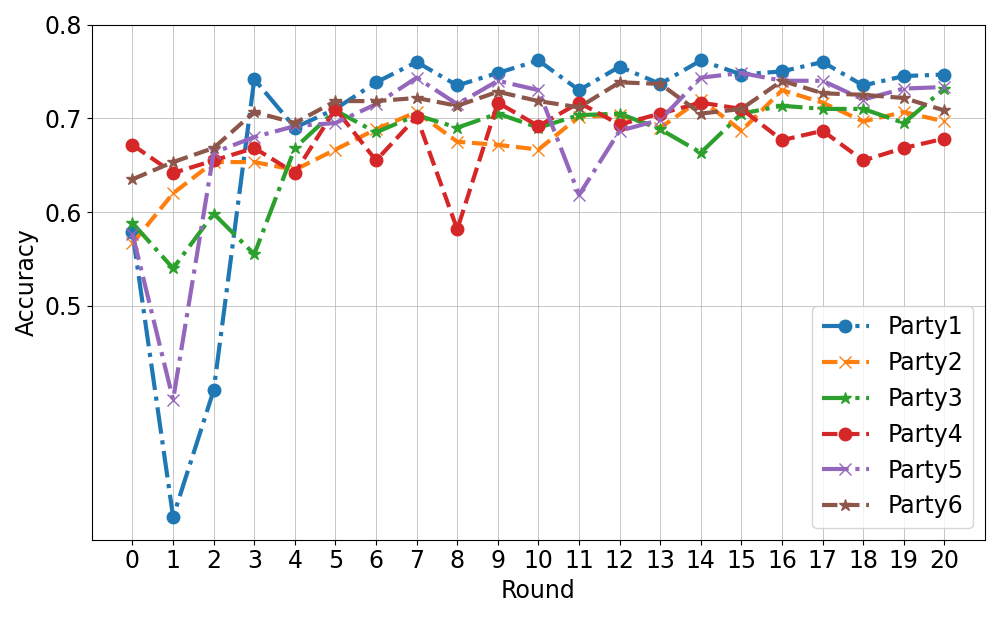}}
\caption{Accuracy v.s. Communication Rounds} \label{fig:epoch}
\end{figure*}

\begin{table*}[t]
\centering
\resizebox{6.350in}{!}{%
\begin{tabular}{l|c|c|c|cll|c|c|c|c}
\cline{1-5} \cline{7-11}
FEMNIST        & k    & $\epsilon$    & $\delta$      & Accuracy &  & CIFAR-10          & k   & $\epsilon$    & $\delta$  & Accuracy \\ \cline{1-5} \cline{7-11} 
FEDMD-NP & 2880 & +$\infty$      & 1          & 96.15\%  &  & FedMD-NP & 300 & +$\infty$      & 1      & 86.88\%  \\ \cline{1-5} \cline{7-11} 
Centralized & 2880 & +$\infty$      & 1          & 98.00\%  &  & Centralized & 300 & +$\infty$      & 1      & 88.83\%  \\ \cline{1-5} \cline{7-11} 
\system\       & 16   & 0.0027 & 0.0062     & 80.64\%  &  & \system\       & 16  & 0.0260 & 0.0583 & 74.40\%  \\ \cline{1-5} \cline{7-11} 
\system\       & 60   & 0.0090 & 0.0206     & 88.06\%  &  & \system\       & 60  & 0.0867 & 0.1815 & 81.58\%  \\ \cline{1-5} \cline{7-11} 
\system\       & 300  & 0.0452 & 0.0989     & 93.56\%  &  & \system\       & 120 & 0.1734 & 0.3301 & 83.57\%  \\ \cline{1-5} \cline{7-11} 
\system\       & 2880 & 0.4342& 0.6321 & 96.63\%  &  & \system\       & 300 & 0.4336 & 0.6327 & 87.38\%  \\ \cline{1-5} \cline{7-11} 
\end{tabular}
}
\caption{Comparisons with All Baselines}\label{tab:compare}
\end{table*}

\subsection{Baselines}
We demonstrate the effectiveness of our proposed \system~by comparison with the following three frameworks. We omit the comparison with FedAvg as it delivers similar utility as the \textit{Centralized} framework.
\begin{enumerate}
    \item \textit{Non-private Federated Model Distillation} (FedMD-NP) framework: all parties train on all their local private data, collaborate the public data distillation, and use the aggregation feedbacks to update the local model as same as FedMD. It should be noted that there is no privacy guarantee in this framework.
    \item \textit{Centralized} framework: 
    the private data of all parties were pooled into a centralized server 
    to train a global model. We use this as an utility upper bound.
    \item \LDP\ framework: \LDP\ requires each party to explicitly add Gaussian noise to locally ensure $(\epsilon,\delta)$-DP before releasing their local model knowledge to the server.
\end{enumerate}

\subsection{Performance Analysis}

\paragraph{Evaluation on privacy budget $\epsilon$.}
It can be observed from Figure \ref{fig:parties}(a) that \system\ can ensure strong privacy protection during training and communication.
For each party in \system\, we fixed $k=3$, which means we only randomly sample three private data points from each private local dataset.
While more parties participate in the training and communication, the number of each private local dataset becomes smaller, since each party's data size equals to the total number of private data $n$ divided by the number of parties $N$.
Due to that, 
when we increase the number of the parties, the privacy budget $\epsilon$ will increase even with a fixed random sample size $k$.
However, 
the $\epsilon$ is still very small, its log 10 scale is close to -2 for CIFAR-10 and -3 for FEMNIST, when the number of parties $N$ is 10.

\paragraph{Evaluation on $\delta$}: Similar to $\epsilon$, $\delta$ is defined based on Theorem \ref{theorem_sampling2}. Figure \ref{fig:parties}(b) shows that the increasing number of parties will increase the $\delta$.
The $\delta$ is small, its log 10 scale is close to -2 for CIFAR-10 and -3 for FEMNIST, when the number of parties $N$ is 10.
Note that, in real life, the private data are collected by each party independently, so more parties would not decrease the local data size in practice.

\paragraph{Evaluation on $\sigma$ in \LDP.}
Besides our proposed approach, the most naive solution is \LDP.
~Unlike \system, Fed-LDP can use all the private dataset for training and only protect the distillation information on the public dataset, as shown in Figure~\ref{fig:CFL_FedMD}(a).
However, FedMD is cursed by a massive number of queries and multi-round communications as per the sequential composition in DP~\cite{dwork2014algorithmic}.
Given the same $ \epsilon $, $\delta$ as in \system, the $\sigma$ is a huge number from 4 to 7 in the log 10 scale, 
compromising the utility of the original information.
Due to this reason, we do not report the experimental results of \LDP\ in this work, since the prediction results are close to random guess with a huge noise scale of $\sigma$ for both FEMNIST and CIFAR-10. 
The only way to maintain utility is to set a very large $\epsilon$ and $\delta$ for \LDP, but it will result in meaningless privacy guarantee.

\paragraph{Evaluation on model convergence.}
Figure \ref{fig:epoch} presents the accuracy trajectories of each party in our \system. As shown in Figure \ref{fig:epoch}, all parties can converge to a decent performance within 20 communication rounds, largely reducing communication cost.
Due to the complexity of the tasks, FEMNIST shows a slightly better performance than CIFAR-10.

\begin{table}[t]
\centering
\resizebox{3.350in}{!}{%
\begin{tabular}{cccccc}
\hline
                      &                          & N  & logits  & softmax & argmax  \\ \hline
\multirow{4}{*}{\system} & \multirow{2}{*}{FEMNIST} & 5  & 74.99\% & 75.02\% & 75.58\% \\ \cline{3-6} 
                      &                          & 10 & 80.74\% & 80.64\% & 81.58\% \\ \cline{2-6} 
                      & \multirow{2}{*}{CIFAR-10}   & 5  & 69.79\% & 70.02\% & 70.01\% \\ \cline{3-6} 
                      &                          & 10 & 74.55\% & 74.88\% & 75.12\% \\ \hline
\end{tabular}
}
\caption{Different Distillation Approaches}\label{tab:distill}
\end{table}

\paragraph{Evaluation on distillation approaches.}
In the original FedMD~\cite{li2019fedmd}, they use logits as the distillation approach. However, in our implementation, besides the logits, we also build the distillation with softmax and argmax approaches.
Softmax approach returns the soft labels and argmax approach returns the hard label for each query.
From Table \ref{tab:distill}, we can see that the results did not differ too much across different approaches in general. However, we recommend argmax label approach 
for both FedMD and our system. There are two main reasons: (1) argmax shows slightly better performance than the other two approaches; (2) more importantly, argmax can save much communication cost of each query. Both softmax and logits need to send the float vectors, but argmax only needs to send the integer during communication.

\paragraph{Evaluation on IID and Non-IID distributions.}
Table \ref{tab:iid} shows the evaluation on Non-IID dataset. 
\system\ can achieve a superior performance with a low privacy cost because of the noise-free differential privacy mechanism.
Compared with IID, non-IID is definitely more challenging due to the incomplete data information of each class.
The detailed settings of our experiments are well introduced in the appendix. 
From the results, we can see that FEMNIST can do better on Non-IID tasks. The main reason is for CIFAR-10 task, we only use one sub-class during training which hardly train the local model well for other classification. 
For example, one party has the wolfs dataset during training, but it is hardly to help classify lions correctly as large carnivores.

\begin{table}[t]
\centering
\resizebox{3in}{!}{%
\begin{tabular}{ccccc}
\hline
&  & N  & IID  & Non-IID  \\ \hline
\multirow{2}{*}{\system} & \multirow{1}{*}{FEMNIST}  
& 10 & 81.58\% & 78.36\%  \\ \cline{2-5} 
& \multirow{1}{*}{CIFAR-10}   
& 10 & 75.12\% & 53.18\%\\ \hline
\end{tabular}
}
\caption{IID vs Non-IID}\label{tab:iid}
\end{table}

\paragraph{Evaluation on number of parities.}
It is not hard to see that more parties can help improve the utility in the federated learning.
Figure \ref{fig:parties}(d) shows that more collaboration can effectively improve the performance of each local model.
Although we only have 10 parties in total, but it already can achieve a good performance on complex image dataset, i.e. CIFAR-10.
Compared with \system, previous private collaborate learning framework requires at least hundreds of parties to be robust to the noise perturbation from previous DP and LDP mechanisms \cite{papernot2017semi,geyer2017differentially,bhowmick2018protection,sun2020ldp}.
In that sense, \system\ also is the work that firstly provides a private collaboration system with a small number of parties.

\paragraph{Comparison on with replacement and without replacement samplings.}
The results in Figure \ref{fig:param} demonstrate the correctness of the lemma \ref{lemma_compare}. It is not hard to see that, while we fix the number of the parties, the increasing number of sampling examples will require a larger privacy budget. Meanwhile, the without replacement sampling spends much more privacy loss than with replacement sampling with the same size of the sampled subset.
Furthermore, we evaluate the performance of both sampling strategies, and the results are shown in Table \ref{tab:compare}. From the results, we find there is no much difference between these two sampling strategies. In summary, we recommend using with replacement sampling for private model training, since it costs less privacy budget than without replacement sampling but achieves the same performance. 

\begin{figure}[t]
\centering
\subfloat[n=100]{\includegraphics[width=1.7in]{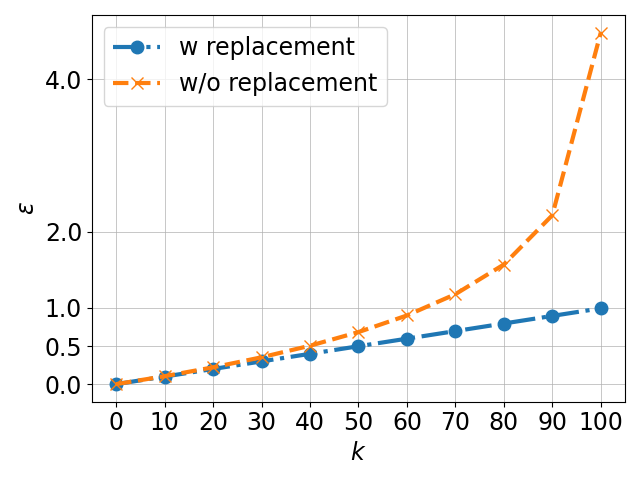}}
\subfloat[n=100]{\includegraphics[width=1.7in]{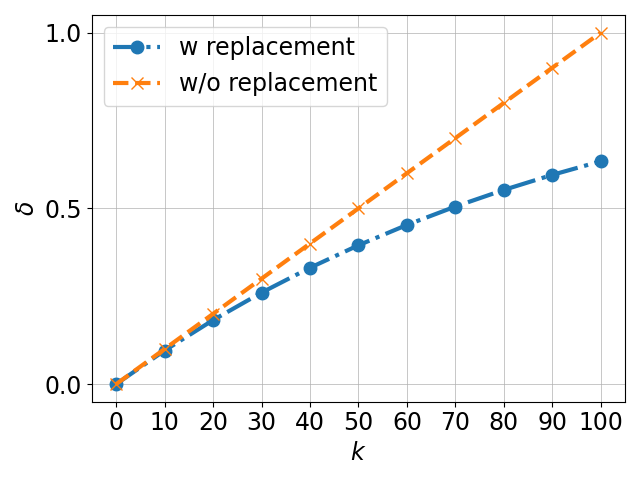}}
\caption{$\epsilon$ and $\delta$ comparison between without (i.e., w/o) replacement and with (i.e., w) replacement. k is the size of the sampled subset.} \label{fig:param}
\end{figure}

\begin{table}[t]
\centering
\resizebox{3.37in}{!}{%
\begin{tabular}{ccccc}
\hline
&  & k  & w replacement  & w/o replacement  \\ \hline
\multirow{2}{*}{\system} & \multirow{1}{*}{FEMNIST}  
& 300 & 93.56\% & 93.63\%  \\ \cline{2-5} 
& \multirow{1}{*}{CIFAR-10}   
& 60 & 81.58\% & 81.13\%\\ \hline
\end{tabular}
}
\caption{With replacement vs without replacement sampling}\label{tab:wo}
\end{table}


\paragraph{Comparison with baselines.}
Table \ref{tab:compare} shows that \system\ can achieve a superior performance with a low privacy cost because of the noise-free differential privacy mechanism.
For all methods, we report the average accuracy of 10 parities.
When we increase the privacy budget, it can even outperform the FEDMD-NP approach and be comparable to the Centralized approach.
Meanwhile, 
we observe that given a very small ($\epsilon$,$\delta$) with $k=16$, we can still achieve a competitive performance on CIFAR-10. None of the previous works related to differential privacy in federated learning \cite{geyer2017differentially,bhowmick2018protection,sun2020ldp} can achieve comparable performance on CIFAR-10 with such a small $\epsilon$ as ours.

Note that, we did not list the utility of the \LDP, since the utility of each model is close to random guess while we use the same privacy budget as \system. Put in another way, \LDP\ requires an extremely large noise scale to ensure the same level of ($\epsilon,\delta$)-DP as \system, resulting in poor utility. However, it no longer becomes a problem in \system. After the random sampling with replacement approach in the first step, \system\ is already $((k\ln{\frac{n+1}{n}}, 1-\left(\frac{n-1}{n}\right)^k)$-differential privacy due to the post-processing property \cite{dwork2014algorithmic}.

\paragraph{Comparison with previous works.}
Our results are competitive comparing to the previous works that adopt CDP and LDP.
Currently, most of the popular differential privacy approaches, such as \DPSGD\ \cite{abadi2016deep}, \PATE\ \cite{papernot2017semi}, are cursed by the number of queries and communication rounds when they are applied to FL.
Since each query touches the private information, 
the large number of queries 
will cost huge privacy budget in FL with multi-round communications.


\section{Discussion}
\paragraph{Diversity of public dataset.}
In federated model distillation, the public dataset requires careful deliberation and even prior knowledge on parties' private datasets. The distribution of the unlabeled public dataset could either match, or differ from the distribution of the private training data available at the parties to some degree. It needs to know how the gap widens when the public dataset becomes more different from the training dataset, the worst case could be from different domains without any overlap.
There is also a potential to use the synthetic data from a pre-trained generator (e.g. GAN) as public data to alleviate potential limitations (e.g. acquisition, storage) of real unlabeled datasets. This may open up numerous possibilities for effective and efficient model distillation.

\paragraph{Diversity of local models.}
\system\ allows local models in FL to not only differ in model structure, size, but also numerical precision, offering great benefit for the Internet of Things (IoT) that involves edge devices with diverse hardware configurations and computing resources.

\paragraph{Weighted aggregation.}
The aggregation step in Algorithm~\ref{alg:FedMD-NF} is based on directly averaging of parties' predictions, i.e., $f_\mathsf{Aggreg}$ corresponds to the average function with equal weight $1/N$, where $N$ is the number of parties. However, parties may contribute to the consensus differently, especially in the extreme cases of model and data heterogeneity. Allocating all the parties with the same weight may negatively impact system utility. We remark that there 
may exist more advanced weighted average algorithms that can further boost utility. These weights can be used to quantify the contributions from local models, and play important roles in dealing with extremely different models.

\paragraph{Limitations.}
Based on the privacy analysis of NFDP mechanisms, for both with replacement and without replacement sampling strategies, we require each local party has an adequate size of the dataset.
For example, if each local data only has one label, our mechanism can not protect any privacy due to the private training data's size limitation.

In this case, NFDP could be very useful for three scenarios.
First, one party is required to provide machine learning as a service (MLaaS) for others, i.e., teacher-student learning framework. While this party contains a large size of private data, NFDP could help it train a privacy guaranteed model.
Second, one party has a large dataset, but the data itself is lack of diversity. The model still can not achieve a good performance due to the data diversity limitation. In this case, they need to communicate with others for a better model utility, and we can use NFDP to protect them during their communications.
Finally, some learning tasks only require a small fraction of the private training data, such as FedMD \cite{vinyals2016matching,snell2017prototypical}.
NFDP can perform well on these tasks with adequate privacy protection.
Besides FedMD, traditional one-shot learning and few-short learning tasks are also suitable to use NFDP for privacy protection for the same reason.

\section{Related Work}

\subsection{Differential Privacy}

Differential privacy \cite{Dwork2006our,dwork2014algorithmic} provides a mathematically provable framework to design and evaluate a privacy protection scheme. 
Recently, differential privacy has been applied to FL~\cite{bhowmick2018protection,geyer2017differentially,mcmahan2018learning}. Previous works mostly focus on the centralized differential privacy mechanism that requires a trusted party~\cite{geyer2017differentially,mcmahan2018learning,yang2021secure}, or local differential privacy, in which each user randomizes its gradients locally before sending it to an untrusted aggregator~\cite{sun2020ldp}, or the hybrid mechanism by combining distributed differential privacy (DDP) with crypto system~\cite{lyu2020lightweight}.

\subsection{Knowledge Distillation}
Knowledge distillation~\cite{bucilu2006model,hinton2015distilling} is originally designed to extract class probability produced by a large DNN or an ensemble of DNNs to train a smaller DNN with marginal utility loss. 
It also offers a powerful tool to share knowledge of a model through its predictions.
Knowledge of ensemble of teacher models has been used to train a student model in previous works~\cite{hamm2016learning,papernot2017semi,wang2019private,sun2020differentially}. For example, Papernot el. al.~\cite{papernot2017semi} proposed PATE, a centralized learning approach that uses ensemble of teachers to label a subset of unlabeled public data in a differentially private manner, then trains a student in a semi-supervised fashion~\cite{dwork2014algorithmic}. We remark that our focus is fundamentally different from the setting of PATE, which requires a trusted aggregator to aggregate the prediction label made by the teacher ensemble and conduct DP mechanisms.

\subsection{Federated Learning}
Federated learning (FL) has emerged as a promising collaboration paradigm by enabling a multitude of parties to jointly construct a global model without exposing their private training data. In FL, parties do not need to explicitly share their training data, they have full autonomy for their local data. FL generally comes in two forms~\cite{mcmahan2017communication}: FedSGD, in which each client sends every SGD update to the server, and FedAVG, in which clients locally batch multiple iterations of SGD before sending updates to the server, which is more communication efficient.

More recently, FedMD~\cite{li2019fedmd} and Cronus~\cite{chang2019cronus} attempted to apply knowledge distillation to FL by considering knowledge transfer via model distillation, in which, the logits on an unlabeled public dataset from parties' models are averaged. In FedMD, each model is first trained on the public data to align with public logits, then on its own private data. In contrast, Cronus mixes the public dataset (with soft labels) and local private data, then trains local models simultaneously. One obvious benefit of sharing logits is the reduced communication costs, without significantly sacrificing utility. However, both works did not offer any theoretical privacy guarantee for sharing model prediction.

\section{Conclusion}
In this work, we formulate a new federated model distillation framework with noise-free differential privacy guarantee for each party. We formally prove that NFDP both with replacement and without replacement sampling can inherently ensure $(\epsilon,\delta)$-differential privacy, eliminating explicitly noise addition and privacy cost explosion issues in the previous works.
Empirical results on various datasets, settings, and heterogeneous model architectures demonstrate that our framework achieves comparable utility by using only a few private samples that are randomly sampled from each party, confirming the effectiveness and superiority of our framework.

In the future, we hope NFDP could support more privacy-preserving machine learning methods, such as semi-supervised learning, pre-training learning, meta-learning, and few-shot learning.
Another direction is that we could optimize the random sampling approach with the advanced data analysis for a better promising and practical privacy guarantee mechanism.
Last but not least, we could use the NFDP mechanism with advanced machine learning methods to support more applications in real life, such as natural language processing, graph analysis, and medical diagnosis.

\bibliography{ref.bib}
\bibliographystyle{named}

\newpage

\appendix

\section*{Appendix}
\label{appendix-theorem_sampling}

In this appendix, we first 
prove the privacy guarantee of random sampling without the replacement and then with the replacement.

\withoutreplacement*
\begin{proof}
$D$ and $D'$ are two neighbouring datasets in the data space $\cD$. $|D|=n$ is the size of the dataset. There are two cases including $D = D' \cup \{u\}$ and $D' = D \cup \{u\}$, where $u$ is the the additional sample.
Let $\cM$ be the random sample mechanism that randomly returns a subset of the data without replacement here.
Let $\cS$ denotes the all subsets in the joint domain of $\cM(D)$ and $\cM(D')$.
Then, we use $\Gamma(D)$, $\Gamma(D')$ denote all subsets of $\cM(D)$ and $\cM(D')$ respectively. $S \in \cS$ is a subset in the domain, where $|S|=k$ denotes the size of the subset.
Then, for a random subset $S$, we have,
\begin{align}
 &   \text{Pr}(\cM(D)=S)= 
    \begin{cases}
     \frac{1}{{|D|\choose k}}, &\text{ if } S \in \Gamma(D), \\
     0, &\text{ otherwise}. 
    \end{cases} \\
&    \text{Pr}(\cM(D^{\prime})=S)= 
    \begin{cases}
     \frac{1}{{|D^{\prime}|\choose k}}, &\text{ if } S \in \Gamma(D'), \\
     0, &\text{ otherwise}. 
    \end{cases}
\end{align}

\noindent \textbf{Case 1 ($D' = D \cup \{u\}$):}
Due to $D \subseteq D'$, then we have,
\begin{align}
& \Pr(\cM(D) \in \Gamma(D)) = 1, \\
& \Pr(\cM(D') \in \Gamma(D)) = \frac{{|D|\choose k}}{{|D'|\choose k}}.
\end{align}
Let $R$ is a random subset of $\cS$ and $R$ is composed by two disjoint subsets, i.e.,
$R = R_D \cup R_{D' \setminus D}$, where $R_D \subseteq \Gamma(D)$ and $R_{D' \setminus D} \in \Gamma(D') \setminus \Gamma(D)$. Then, we have 
\begin{align}
& \text{Pr}(\cM(D)\in R) \\
=& \text{Pr}(\cM(D)\in R_{D}) + \text{Pr}(\cM(D)\in R_{D'\setminus D}) \\
=& \text{Pr}(\cM(D)\in R_{D}) + 0 \\
=& \text{Pr}(\cM(D) \in R_{D}) \\
=& \text{Pr}(\cM(D') \in R_{D}) \cdot \frac{{|D^{\prime}| \choose k}}{{|D| \choose k}} \\
=& \text{Pr}(\cM(D^{\prime})\in R_{D})\cdot \frac{{n+1 \choose k}}{{n \choose k}}  \\
=& \text{Pr}(\cM(D^{\prime})\in R_{D})\cdot \frac{n+1}{n+1-k} \\
\leq& \text{Pr}(\cM(D^{\prime})\in R)\cdot \frac{n+1}{n+1-k} \\
\end{align}

\noindent \textbf{Case 2 ($D = D' \cup \{u\}$):}
Due to $D' \subseteq D$, then we have,
\begin{align}
& \Pr(\cM(D) \in \Gamma(D')) = \frac{{|D'|\choose k}}{{|D|\choose k}}, \\
& \Pr(\cM(D') \in \Gamma(D')) = 1.
\end{align}
Let $P$ is a subset of $\Gamma(D) \setminus \Gamma(D')$, then we have
\begin{align}
\Pr(\cM(D) \in P) & \leq \Pr(\cM(D) \in \Gamma(D) \setminus \Gamma(D')), \\
                    &\leq 1-\frac{{n-1 \choose k}}{{n \choose k}}=\frac{k}{n} .
\end{align}

Let $R$ is a random subset of $\cS$ and $R$ is composed by two disjoint subsets, i.e.,
$R = R_{D'} \cup R_{D \setminus D'}$, where $R_{D'} \subseteq \Gamma(D')$ and $R_{D \setminus D'} \subseteq \Gamma(D) \setminus \Gamma(D')$. Then, we have 
\begin{align}
& \text{Pr}(\cM(D)\in R) \\
=& \text{Pr}(\cM(D)\in R_{D'}) + \text{Pr}(\cM(D)\in R_{D\setminus D'}) \\
\leq& \text{Pr}(\cM(D)\in R_{D'}) + \frac{k}{n} \\
\leq& \text{Pr}(\cM(D') \in R_{D'}) \cdot \frac{{|D'| \choose k}}{{|D| \choose k}} + \frac{k}{n}\\
\leq& \text{Pr}(\cM(D^{\prime})\in R_{D'})\cdot \frac{{n-1 \choose k}}{{n \choose k}}  + \frac{k}{n}\\
\leq& \text{Pr}(\cM(D^{\prime})\in R_{D'})\cdot \frac{n-k}{n} + \frac{k}{n}\\
\leq& \text{Pr}(\cM(D^{\prime})\in R)\cdot \frac{n-k}{n} + \frac{k}{n}\\
\end{align}

Now, we merge the Case 1 and 2 together. Then we have $e^\epsilon=\max(\frac{n+1}{n+1-k}, \frac{n-k}{n})=\frac{n+1}{n+1-k}$ and $\delta=\max(0, \frac{k}{n})=\frac{k}{n}$.
Therefore, NFDP without replacement statisfies $(\ln{\frac{n+1}{n+1-k}}, \frac{k}{n})$-differential privacy.
\end{proof}

\withreplacement*
\begin{proof}
Here we use the same notation as the last proof.
The proof of replacement is almost similar to the without replacement.
First, for a random subset $S \in \cS$, we have 
\begin{align}
 &   \text{Pr}(\cM(D)=S)= 
    \begin{cases}
     \frac{1}{|D|^k}, &\text{ if } S \in \Gamma(D), \\
     0, &\text{ otherwise}. 
    \end{cases} \\
&    \text{Pr}(\cM(D^{\prime})=S)= 
    \begin{cases}
     \frac{1}{|D^{\prime}|^k}, &\text{ if } S \in \Gamma(D'), \\
     0, &\text{ otherwise}. 
    \end{cases}
\end{align}

\noindent \textbf{Case 1 ($D' = D \subseteq \{u\}$):}
Due to $D \in D'$, then we have,
\begin{align}
& \Pr(\cM(D) \in \Gamma(D)) = 1, \\
& \Pr(\cM(D') \in \Gamma(D)) = \frac{|D|^k}{{|D'|^k}}.
\end{align}
Let $R$ is a random subset of $\cS$ and $R$ is composed by two disjoint subsets, i.e.,
$R = R_D \cup R_{D' \setminus D}$, where $R_D \subseteq \Gamma(D)$ and $R_{D' \setminus D} \in \Gamma(D') \setminus \Gamma(D)$. Then, we have 
\begin{align}
& \text{Pr}(\cM(D)\in R) \\
=& \text{Pr}(\cM(D)\in R_{D}) + \text{Pr}(\cM(D)\in R_{D'\setminus D}) \\
=& \text{Pr}(\cM(D)\in R_{D}) + 0 \\
=& \text{Pr}(\cM(D) \in R_{D}) \\
=& \text{Pr}(\cM(D') \in R_{D}) \cdot \frac{{|D^{\prime}|^k}}{{|D|^k}} \\
=& \text{Pr}(\cM(D^{\prime})\in R_{D})\cdot \frac{{(n+1)^k}}{{n^k}}  \\
\leq& \text{Pr}(\cM(D^{\prime})\in R)\cdot \left(\frac{{n+1}}{{n}}\right)^k \\
\end{align}

\noindent \textbf{Case 2 ($D = D' \cup \{u\}$):}
Due to $D' \subseteq D$, then we have,
\begin{align}
& \Pr(\cM(D) \in \Gamma(D')) = \frac{{|D'|^k}}{|D|^k}, \\
& \Pr(\cM(D') \in \Gamma(D')) = 1.
\end{align}
Let $P$ is a subset of $\Gamma(D) \setminus \Gamma(D')$, then we have
\begin{align}
\Pr(\cM(D) \in P) & \leq \Pr(\cM(D) \in \Gamma(D) \setminus \Gamma(D')), \\
                    &\leq 1-\left(\frac{n-1}{n}\right)^k.
\end{align}

Let $R$ is a random subset of $\cS$ and $R$ is composed by two disjoint subsets, i.e.,
$R = R_{D'} \cup R_{D \setminus D'}$, where $R_{D'} \subseteq \Gamma(D')$ and $R_{D \setminus D'} \subseteq \Gamma(D) \setminus \Gamma(D')$. Then, we have 
\begin{align}
& \text{Pr}(\cM(D)\in R) \\
=& \text{Pr}(\cM(D)\in R_{D'}) + \text{Pr}(\cM(D)\in R_{D\setminus D'}) \\
\leq& \text{Pr}(\cM(D)\in R_{D'}) + \left(1-\left(\frac{n-1}{n}\right)^k\right) \\
\leq& \text{Pr}(\cM(D') \in R_{D'}) \cdot \frac{{|D'|^k}}{{|D|^k}} + \left(1-\left(\frac{n-1}{n}\right)^k\right)\\
\leq& \text{Pr}(\cM(D^{\prime})\in R_{D'})\cdot \frac{{(n-1)^k}}{{n^k}}  + \left(1-\left(\frac{n-1}{n}\right)^k\right)\\
\leq& \text{Pr}(\cM(D^{\prime})\in R)\cdot \left(\frac{n-1}{n}\right)^k + \left(1-\left(\frac{n-1}{n}\right)^k\right)\\
\end{align}

Now, we merge the Case 1 and 2 together. Then we have $e^\epsilon=\max(\left(\frac{n+1}{n}\right)^k, \left(\frac{n-1}{n}\right)^k)=\left(\frac{n+1}{n}\right)^k$ and $\delta=\max(0, 1-\left(\frac{n-1}{n}\right)^k)=1-\left(\frac{n-1}{n}\right)^k$.
Therefore, NFDP with replacement statisfies $(k\ln{\frac{n+1}{n}}, 1-\left(\frac{n-1}{n}\right)^k)$-differential privacy.
\end{proof}

\compare*
\begin{proof}
Sampling with replacement is $(k\ln{\frac{n+1}{n}}, 1-\left(\frac{n-1}{n}\right)^k)$-differential privacy and sampling without replacement is $(\ln{\frac{n+1}{n+1-k}}, k/n)$-differential privacy.
Let $n \geq 1$, and then if $k=0$ or $k=1$,
\begin{align}
  \epsilon: & k\ln{\frac{n+1}{n}} = \ln{\frac{n+1}{n+1-k}} \\
    & \text{and} \\
    \delta: & 1-\left(\frac{n-1}{n}\right)^k = k/n
\end{align}

If $1 < k \leq n$, we have
\begin{align}
  \epsilon: & k\ln{\frac{n+1}{n}} < \ln{\frac{n+1}{n+1-k}} \\
    & \text{and} \\
    \delta: & 1-\left(\frac{n-1}{n}\right)^k < k/n
\end{align}
Briefly, we can prove above two inequalities by mathematical induction. First, if $k = 2$, above two inequalities are correct. Then we assume the inequalities are correct while $k=n-1$. Last, we easily prove if $k=n$, the above two inequalities are still correct.
Therefore, for any fixed $n$, $0 \leq k \leq n$, sampling with replacement is more private than Sampling without replacement.
\end{proof}

\end{document}